\documentclass[12pt,a4paper,oneside]{article}

\usepackage{amsmath,amssymb,mathtools,amsthm,amscd,rotating,array,stmaryrd}
\usepackage[authoryear,round]{natbib}
\usepackage{enumitem,multirow}
\usepackage{psfrag,epsfig,boxedminipage}
\usepackage{pgfplots}
\usepackage{dsfont}
\pgfplotsset{compat=1.11}

\newtheorem*{theorem*}{Argmax Theorem}
\newtheorem{corollary}{Corollary}
\newtheorem{theorem}{Theorem}
\newtheorem{definition}{Definition}
\newtheorem{lemma}{Lemma}

\renewenvironment{proof}{{\bfseries Proof.}}

\theoremstyle{definition}

\newcommand{\NN}{\mathbb{N}}

\newcommand{\E}{\mathrm{E}}      
     
\newcommand{\D}{s}

\begin{document}

\title{Truncating the Exponential with a Uniform Distribution}
\author{Rafael Wei{\ss}bach \& Dominik Wied \\[2mm] \textit{\footnotesize{Chairs of Statistics and Econometrics,}}   \\[-2mm]
        \textit{\footnotesize{  Faculties for Economic and Social Sciences,}} \\[-2mm]
        \textit{\footnotesize{University of Rostock \& University of Cologne}} \\
        }
\date{ }
\maketitle

\renewcommand{\baselinestretch}{1.5}\normalsize

\begin{abstract}
For a sample of Exponentially distributed durations we aim at point estimation and a confidence interval for its parameter. A duration is only observed if it has ended within a certain time interval, determined by a Uniform distribution. Hence, the data is a truncated empirical process that we can approximate by a Poisson process when only a small portion of the sample is observed, as is the case for our applications. We derive the likelihood from standard arguments for point processes, acknowledging the size of the latent sample as the second parameter, and derive the maximum likelihood estimator for both. Consistency and asymptotic normality of the estimator for the Exponential parameter are derived from standard results on M-estimation. We compare the design with a simple random sample assumption for the observed durations. Theoretically, the derivative of the log-likelihood is less steep in the truncation-design for small parameter values, indicating a larger computational effort for root finding and a larger standard error. In applications from the social and economic sciences and in simulations, we indeed, find a moderately increased standard error when acknowledging truncation.   	
 \\[2mm]
\end{abstract}

\section{Introduction}
Poor sample selection is a frequent basis for objection to the inferential quality of data. Hospital controls may be negatively selective, a student sample is a positive selection. Sampling from soldiers is selective, because a body height threshold truncates smaller recruits. Inference from the status quo of a loan portfolio can take into account the fact that earlier loan applications with too small score had been rejected \cite[see][]{Buecker}.
Here we study de-selection on the basis of age being either too low or too high. An age is the duration between two events, denoted as ``birth'' and ``death'', and Figure \ref{exa}(left) shows the three possible situations.  
\setlength{\unitlength}{1cm}
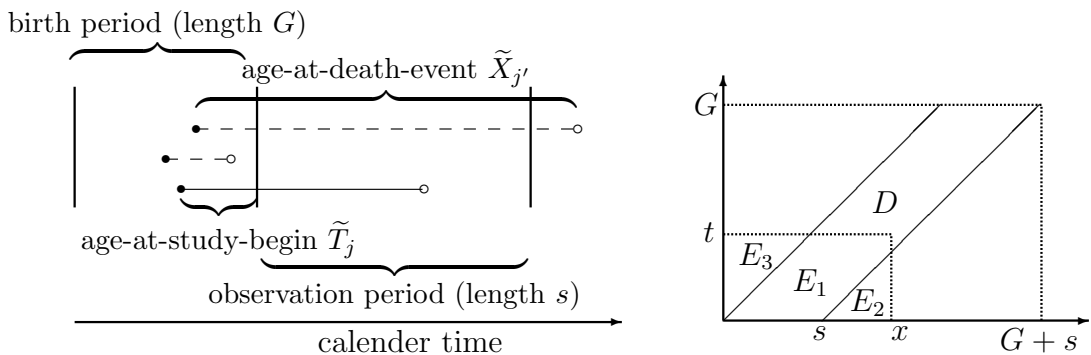
\begin{figure}[htb!] \centering
	\setlength{\unitlength}{0.8cm}
	\begin{picture}(10,5.5)
	\linethickness{0.3mm}
	\put(0,0.3){\vector(1,0){9.0}} \put(4.0,-0.2){calender time}
	\put(0,2.2){\line(0,1){2}}  
	\put(-1.1,4.6){$\overbrace{\vphantom{ } \hspace{2.3cm}}^{\mbox{\small birth period (length $G$)}}$}
	\put(3,2.2){\line(0,1){2}} 
	\put(7.5,2.2){\line(0,1){2}} 
	
	\linethickness{0.15mm}
	
	\put(2,3.5){\circle*{0.15}} \put(8.275,3.5){\circle{0.15}} 
	\multiput(2,3.5)(0.4,0){16}{\line(1,0){0.2}}
	
	\put(1.5,3){\circle*{0.15}} \put(2.575,3){\circle{0.15}} 
	\multiput(1.5,3)(0.4,0){3}{\line(1,0){0.2}}

	\put(1.75,2.5){\circle*{0.15}} \put(5.75,2.5){\circle{0.15}} 
	\put(1.625,2.5){ \line(1,0){3.925}}

	\put(2.0,3.8){$\overbrace{\vphantom{ } \hspace{5.0cm}}^{\mbox{\small age-at-death-event $\widetilde{X}_{j'}$}}$} 
	\put(0.1,2.4){$\underbrace{\vphantom{ } \hspace{1.0cm}}_{\mbox{\small age-at-study-begin $\widetilde{T}_j$}}$} 
	
	\put(2.2,1.4){$\underbrace{\vphantom{ } \hspace{3.5cm}}_{\mbox{\small observation period (length $s$)}}$} 
	\end{picture}
	\setlength{\unitlength}{0.13cm}
	\begin{picture}(40,30)
	\linethickness{0.3mm}
	\put(2,2){\vector(1,0){37}} \put(11.0,0){$s$} \put(19.0,0){$x$} \put(30.0,-1){$G+s$} 
	\put(2,2){\vector(0,1){25}} \put(0,10){$t$} \put(-1,23){$G$}  
	
	\put(2,2){\line(1,1){22}}
	\put(12,2){\line(1,1){22}}
	
	\multiput(2,24)(0.4,0){81}{\line(1,0){0.2}}
	\multiput(34.2,2)(0,0.4){55}{\line(0,1){0.2}}
	\multiput(2,10.8)(0.4,0){43}{\line(1,0){0.2}}
	\multiput(19.0,2)(0,0.4){22}{\line(0,1){0.2}}

	\put(14.8,3.0){$E_2$}	
	\put(9,5.0){$E_1$}
	\put(3.5,7.7){$E_3$}	
	\put(17,13.0){$D$}

	\end{picture}
	\caption{Left: Three cases of the date of 1$^{st}$ event (black bullet) and date of 2$^{nd}$ event (white circle): observed (solid) and truncated (dashed) durations/ Right: Sets in the co-domain of $(X_i,T_i)'$ or $(\widetilde{X}_j, \widetilde{T}_j)'$ used in Lemma \ref{lemma_imeasure} (and proof): Example $x \ge s$}
	\label{exa} (Explanation of panels and symbols is distributed over larger parts of text.)
\end{figure}

We assume an Exponential distribution for the latent duration $\widetilde{X}_j$, observed or truncated, and estimate its parameter $\theta_0$. Our three applications will be the lifetime of a company (in Germany), the duration of a marriage (in the city of Rostock), and the waiting time, after the 50$^{th}$ birthday, until dementia onset (in Germany).

The parameter of an Exponential distribution is closely linked to the probability of the second event happening within one time unit, one year in all of our applications. In essence, one wants to estimate such an event probability by dividing the number of events (over a certain period) by the number of units at risk (at the beginning of the period), this being prohibited by the lack of denominator. We circumvent the missing data with the conditional distribution of the duration. 

We distinguish, as three statistical masses, the {\it population} as all units with a first event in a period (of length $G$), the latent {\it simple random sample} and, after truncation, the {\it observed data}.

One can of course ask, in particular whether such simple random latent samples exist at all in practice. In survival analysis, the assumption of durations as independent identically distributed random variables can be defended, because independence and randomness are attributable to an unforeseeable staggered entry \cite[see e.g.][]{A-likeliho:2009}. Even more specifically, in labour economics, it is validated theoretically that market friction renders the entry into a new occupation for an employee random, and hence its duration until the new occupation.

Truncation is known to introduce a selection bias, referring to the comparison of {\em two} models, the estimate of the correct model compared to the estimate from erroneously modelling the observed data as a simple random sample (srs-design). (We will later distinguish the {\it selection bias} from the {\it statistical bias}, the later referring to only {\em one} model, namely comparing an estimate with the true parameter.)  More important for us is that truncation is suspected to increase the standard error, as suspected by \cite{AdTa2019} due to dependence in the observed data, and we are interested in the extent to which the truncation hinders statistical inference in terms of large sample properties. 

 As an early reference, \cite{cox1974} in their Example 2.25 consider the size of the truncated sample as an ancillary statistic, not acknowledging the size of the latent sample, $n$, as a parameter. The size of the truncated sample was subsequently considered again as random in \cite{woodroofe1985}, and conditioning was used to prove consistency. 
 Neighbouring contemporaneous work on truncation in survival analysis, mostly semi- and non-parametrically are \cite{shen2010,moreira2010a,weisponi2012,emura2015,emura2017,franchae2019,doerre2020}.

 Here, we derive the maximum likelihood estimator (MLE) of $n$ and $\theta_0$ by representing the observed data as a truncated empirical process. We derive the likelihood with standard results for empirical processes \cite[see e.g.][]{reiss1993}. The size of the data $m$ will be shown to be such a process, seen as a point measure, evaluated at a certain set $S$.  To the best of our knowledge the model is the first example of an exponential family with the space of point measures being the sample space.

\section{Model and Result}

Before presenting the estimator and its asymptotic distribution, the data need to be described. 

\subsection{Sample Selection}

The unit $j$ of the latent sample carries as second measure its lifetime $\widetilde{X}_j$ ($\in \mathbb{R}_0^+$) its birthdate (a calendar time). We, equivalently, measure the birthdate {\em backwards} from a specific time point (equal for all units of the latent sample) and denote it as $\widetilde{T}_j$. We use the calendar date when our study period begins as thus time point, so that $\widetilde{T}_j$ has the interpretation of being the ``age when the study begins''. We consider as population, units born within a pre-defined time window going back $G$ time units from the study beginning, so that $\widetilde{T}_j \in [0,G]$ (see Figure \ref{exa}(left)). Define $S:=\mathbb{R}_0^+ \times [0,G]$, with $0<G < \infty$, the space for one outcome, and let it generate the $\sigma$-field $\mathcal{B}$.  In comparison to the example of soldiers whose recruitment truncates {\it all} at the {\it same} height, to fit our survival analytic applications, each unit is truncated at a {\it different} age. As illustrated in Figure \ref{exa}(left), all units are truncated at the same {\it time}, when the study begins. Differently for each unit $j$, the time interval of observation truncates the sample unit in cases of a too low or too high age. Because $\widetilde{T}_j$ is the (shifted) birth date, assuming as births process a time-homogeneous Poisson process renders the distribution of $\widetilde{T}_j$ to be Uniform \cite[see][Lemma 2]{doerre2020}. 
Let us collect the following notation and assumptions:
\begin{enumerate}[label= \textnormal{(A\arabic{*})}, leftmargin=1cm]
	\item\label{A1:Compact} Let $\Theta:=[ \varepsilon, 1 / \varepsilon]$ for some ``small'' $\varepsilon \in ]0,1[$.
	\item\label{A2:SquareInt} Let for $\theta_0$ being an interior point of $\Theta$, $\widetilde{X}_j \sim Exp(\theta_0)$, i.e. with density $f_E(\cdot/\theta_0)$ and CDF $F_E(\cdot/\theta_0)$ of the Exponential distribution. Let $\widetilde{T}_j \sim Uni[0,G]$, with density $f^{\widetilde{T}}$ and CDF $F^{\widetilde{T}}$ of the Uniform distribution.
	\item\label{A3:Ind} $\widetilde{X}_j$ and $\widetilde{T}_j$ are stochastically independent.
	\item\label{A4:trunc} For known constant $\D>0$, column vector $(\tilde{X}_j,\tilde{T}_j)'$ is observed if it is in 
	\[
	D:=\{ (x,t)' | 0 < t \leq x \leq t + \D, t \le G \}.
	\] 
\end{enumerate}
Assumption \ref{A4:trunc} formalises that a sample unit is only observed when its second event falls into the observation period (of length $s$). For instance, in one of the applications, we will know the age-at-insolvency, i.e. the duration until insolvency, only for those companies that filed for insolvency within the $s=3$ years 2014 -- 2016. 
The parallelogram $D$ is depicted in Figure \ref{exa}(right). Following up on \ref{A4:trunc}, we denote an {\it observation} by $(X_i,T_i)'$, $i=1, \ldots, m \le n$. 

 The paper assumes a simple random sample for $(\widetilde{X}_j,\widetilde{T}_j)'$, $j=1,\ldots,n, n \in \mathbb{N}$, i.e. $\mbox{i.i.d.}$ random variables (r.v.) mapping from the probability space $(\Omega,\mathcal{A},P)$ onto the measurable space $(S, \mathcal{B})$. 

Define now for $\theta \in \Theta$
  \begin{equation} \label{alpha0}
 	\alpha_{\theta}  := \frac{(1- e^{-\theta s})(1- e^{-G  \theta})}{G \theta}
 \end{equation}
and note that for $\theta=\theta_0$, by Figure \ref{exa}(right), Fubini's Lemma and the substitution rule, it is $P\{\widetilde{T}_j \le \widetilde{X}_j \le \widetilde{T}_j +s\}$, i.e. the selection probability of the $j^{\mathrm{th}}$ individual. The numerator is, due to $\theta_0,s,G >0$, strictly positive and, as to be expected, with a larger observation interval, i.e. increasing $s$, the selection becomes more likely. Additionally, for larger $\theta_0$ (or smaller expected duration) the denominator increases faster than the numerator does, so that the selection becomes less likely. A shorter duration will not reach the observation interval. Seen as a function of $G$, $\alpha_{\theta_0}$ is monotonously decreasing, with almost the same interpretation. 

The selection probability will occur in the likelihood, so that for maximisation, its first derivative will be needed. The second derivative of $\alpha_{\theta}$ (with now variable $\theta$) will be needed for proving the asymptotic normality and thus calculating the standard error. The proof is elementary and omitted here.   

\begin{corollary}\label{delalpha} With Assumptions \ref{A1:Compact}-\ref{A4:trunc} the first and second derivatives of \eqref{alpha0} in $\theta$ are: 
	\begin{eqnarray*}
 \dot{\alpha}_{\theta} & = & 	\frac{ \theta s e^{-\theta s} (1- e^{-G  \theta })
	+ \theta (1- e^{-\theta s}) \, G \, e^{-G  \theta } - (1- e^{-\theta s})(1- e^{-G  \theta})}{G \theta^2} \\
 \ddot{\alpha}_{\theta}  & = & e^{- \theta s} \left( - \frac{2 s}{G \theta^2} - \frac{s^2}{G \theta} - \frac{2}{G \theta^3} \right) 
  + e^{- G \theta} \left(-\frac{2}{\theta^2} - \frac{G}{\theta} - \frac{2}{G \theta^3} \right)	\\
 & & + 	e^{- (G+s) \theta} \left( \frac{2 s +G}{G \theta^2} + \frac{(G+s)s}{G \theta} + \frac{1}{\theta^2} + \frac{G+s}{\theta} + \frac{2}{G \theta^3} \right) - \frac{2}{G \theta^3}
 	\end{eqnarray*}
\end{corollary}

Obviously, the distribution of $(\widetilde{X}_j,\widetilde{T}_j)'$, conditional on being observed, will become important for deriving the likelihood. 
\begin{definition} \label{defxt}
Let	$(X_1,T_1)'$, $(X_2,T_2)'$, $(X_3,T_3)'$, \ldots, be independent and identically distributed with CDF
	\[
		F^{X,T}(x,t)= P\left\{\widetilde{X}_j \le x,\widetilde{T}_j \le t|\widetilde{T}_j \leq \widetilde{X}_j \leq \widetilde{T}_j + \D\right\}.
	\]
\end{definition}

In more detail, the distributions of $(X_i,T_i)'$ and $X_i$ will be needed on the one hand later to define the precise stochastic description of the data, i.e. of the truncated sample as a truncated empirical process. On the other hand, we already need the distribution (and also moments) here to study the consistency and asymptotic normality of the maximum likelihood estimator. The proofs of Lemma \ref{lemma_imeasure} and Corollary \ref{EXi} are elementary (and omitted), but it is useful to define sets (see Figure \ref{exa}(right)):
\begin{equation*}
\begin{split}
E_1 & := [0,x] \times [0,t] \cap D, \; E_3  := \text{triangle spanned by points} \; (0,0)',(0,t)',(t,t)', \\
E_2 & :=  \text{triangle spanned by points} \; (s,0)',(x,0)',(x,x-s)' \quad (if \; x \ge s,  \; else \; \emptyset)
\end{split}
\end{equation*}
 
\begin{lemma} \label{lemma_imeasure} With Definition \ref{defxt} and under Assumptions \ref{A1:Compact}-\ref{A4:trunc} it holds, for $(x,t)' \in D$, $\alpha_{\theta_0} F^{X,T}(x,t)  =  (1- e^{- \theta_0 x}) t/G  - R(x,t)$, with $\frac{\partial^2}{\partial x \partial t} R(x,t) =0$.
\end{lemma}

\begin{corollary}\label{EXi}
 With Definition \ref{defxt} and under Assumptions \ref{A1:Compact}-\ref{A4:trunc}:  
	\begin{itemize}
		\item[(i)] For $(x,t)' \in D$ it holds $f^{X,T}(x,t)= \frac{\theta_0}{G \alpha_{\theta_0}} e^{- \theta_0 x}$ (outside $D$ it is zero).
		\item[(ii)] The marginal density of $X$ for $x \in [0,G+s]$ is 
		\begin{eqnarray*}
		f^X(x) & = &  \frac{\theta_0}{G \alpha_{\theta_0}} e^{- \theta_0 x} \left(\mathds{1}_{[0,s]}(x) x  + \mathds{1}_{]s,G]}(x) s  + \mathds{1}_{]G,G+s]}(x) (G-x+s) \right).
		\end{eqnarray*}
		\item[(iii)] For the expectation of $X_i$ it holds
		\[
		\alpha_{\theta_0} E_{\theta_0}(X_i)=A(s,G,\theta_0) e^{-\theta_0 s} + B(s,G,\theta_0) e^{-\theta_0 G} + C(s,G,\theta_0) e^{-\theta_0 (G+s)} + \frac{2}{G\theta_0^2},
		\]
		with $A(s,G,\theta_0)  :=  -\frac{s}{G\theta_0}- \frac{2}{G\theta_0^2}, B(s,G,\theta_0) :=  - \frac{1}{\theta_0} -  \frac{2}{G \theta_0^2}$ and $C(s,G,\theta_0)  :=  \frac{G+s}{G \theta_0} +  \frac{2}{G \theta_0^2}$. 
		\item[(iv)] For the variance of $X_i$ note that
			\[
		E_{\theta_0}(X^2_i)=A^q(s,G,\theta_0) e^{-\theta_0 s} + B^q(s,G,\theta_0) e^{-\theta_0 G} + C^q(s,G,\theta_0) e^{-\theta_0 (G+s)} + \frac{1}{4\theta_0^3}
		\]
		with $A^q(s,G,\theta_0)= \frac{-s^2}{G \theta_0} - \frac{s}{6 \theta_0^2} - \frac{1}{4 \theta_0^3}$, $B^q(s,G,\theta_0)=\frac{-G}{\theta_0} - \frac{4}{\theta_0^2} - \frac{1}{4 \theta_0^3}$ and $C^q(s,G,\theta_0)= \frac{(G+s)^2}{G \theta_0} + \frac{G+s}{6 \theta^2} + \frac{1}{4 \theta_0^3}$.		
	\end{itemize}
\end{corollary}

We are now in the position to formulate the likelihood, maximise it and apply large sample theory.

\subsection{Estimator and Confidence Interval}

Similar to $P\{A \cap B\}=P\{A|B\}P\{B\}$ and with detailed definitions following, we decompose the density of the observations and the random sample size, i.e. the likelihood $\ell$, into the product of the conditional density of the data - conditional on observation - and the distribution of the observation count. If the observations - conditional on having been observed - are independent, the first factor of such product, again, is a product, namely over the conditional densities of each observation.  

W.r.t. the second of such factors, note that the size of the observed sample has a Binomial distribution. We can approximate it by a Poisson distribution, when - as is usually argued with the probability generating function - the selection  probability $\alpha_{\theta}$  for each of the $n$ i.i.d. latent Bernoulli experiments is small. This is the case when the width of the observation period (of length $s$) is ``short'', relative to population period (of length $G$), as will be true for our applications. 
The description so far motivates
	\begin{equation}\label{likeunsauber} 
	\ell(m, \sum_{i=1}^m x_i;\theta,n) \approx \theta^m \exp\left(-\theta \sum_{i=1}^m x_i\right)  n^m \exp (- n \alpha_{\theta}),
\end{equation}
where we already use the ``generic'' parameter $\theta$, as will be explained at the end of Section \ref{theory}.
The conditionally independent and Exponentially distributed observed durations $X_i$ cause the first two factors in  \eqref{likeunsauber}. The last two factors appear in the Poisson distribution of the observed sample size with parameter $n \alpha_{\theta}$. Details for the likelihood construction will need a formulation of the data as truncated empirical process and will be given in Section \ref{theory} (and in Theorem \ref{theolike}). The main topic is that it is not necessary to formulate the conditional independence as further assumptions, but that it follows from the simple sample assumption for the $(\widetilde{X}_j,\widetilde{T}_j)'$ and Assumptions \ref{A1:Compact}-\ref{A4:trunc}. At first reading, Section \ref{theory} may be omitted  without lack of coherence. 
 
As a side remark, by inspection of \eqref{likeunsauber}, and long-known for random left-truncated durations, the likelihood does not include the (observed) $t_i$, but it does include the (unobserved) $n$. Accordingly $n$, that has not been a parameter in the model \ref{A1:Compact}-\ref{A3:Ind}, becomes a parameter after adding \ref{A4:trunc}.

As $n$ is unknown in likelihood \eqref{likeunsauber} (and equally in its rigorous counterpart to follow in Theorem \ref{theolike}), we obtain the approximate MLE for $(n,\theta_0)$ and use the $\theta$-coordinate of the bivariate zero as $\hat{\theta}$. 
The logarithm of the likelihood has the derivative 
\begin{equation} \label{scoretheta}
	\frac{\partial}{\partial \theta} \log \ell\left(m, \sum_{i=1}^m x_i; \theta,n\right)  =  \frac{m}{\theta}  -  \sum_{i=1}^m x_i  
	- n  \dot{\alpha}_{\theta}.
\end{equation}
Solving the bivariate equation for $n \in \mathbb{R}^+$ results in $m/\alpha_{\theta}$. 
In order to facilitate the proofs later on, we formulate the estimation as a minimization problem, and in detail as a minimization of an average. Define  
\begin{equation} \label{defpsi}
	\begin{split}
		\psi_{\theta}(\tilde x_j,\tilde t_j) & :=  i_j \left(\tilde x_j - \frac{1}{\theta} + \frac{\dot{\alpha}_{\theta}}{\alpha_{\theta}}  \right)  
		=  i_j \left( \tilde x_j - \frac{1}{\theta} \right. \\  + &
		\left. \frac{ \theta s e^{-\theta s} (1- e^{-G  \theta })
			+ \theta (1- e^{-\theta s}) \, G \, e^{-G  \theta } - (1- e^{-\theta s})(1- e^{-G  \theta})}{\theta (1- e^{-\theta s})(1- e^{-G  \theta})}
		\right),
	\end{split}
\end{equation}
with $i_j$ as a realization of $I_j := \mathds{1}_{[\tilde T_j,\tilde T_j + s]}(\tilde X_j)$.

The derivative of the log-likelihood is now obviously related to \cite[see][Sect. 5]{vaart1998} 
\begin{equation} \label{z_estimator}
	\Psi_n(\theta):= \frac{1}{n} \sum_{j=1}^n \psi_{\theta}(\widetilde{X}_j,\widetilde{T}_j).
\end{equation}
The function {\em is} not observable, but it {\em becomes} observable after multiplication by $n$ and hence its zero, $\hat{\theta}$, is observable.

In order to account for boundary maxima, define the MLE $\hat{\theta}$  now as the zero of $\Psi_n(\theta)$ if it exists in (the open) $\Theta$, as $\varepsilon$ if $\Psi_n(\theta)>0$, respectively as $1/\varepsilon$ if $\Psi_n(\theta)<0$, both for all $\theta \in \Theta$.  The following analytical properties (with proof in Appendix \ref{prooflem2}) will be needed to prove the consistency and asymptotic normality of $\hat{\theta}$.

\begin{lemma}\label{allebeding} Under the Assumptions \ref{A1:Compact}-\ref{A4:trunc} it is 
	\begin{itemize}
		\item[(i)] $\psi_{\theta}(\tilde{x}_j,\tilde{t}_j)$ twice continuously differentiable in $\theta$ for every $(\tilde x_j, \tilde t_j)'$,
		\item[(ii)]  for $(\tilde x_j, \tilde t_j)' \in D$
		\begin{equation} \label{psiklein0}	
			\dot{\psi}_{\theta}(\tilde{x}_j,\tilde{t}_j)=  i_j 
			\left(\frac{2}{\theta^2} - 
			\frac{s^2 e^{-\theta s}}{(1- e^{- \theta s})^2} - \frac{G^2 e^{-G \theta}}{(1- e^{- G \, \theta})^2}  \right) > 0,
		\end{equation}	 
		\item[(iii)] $E_{\theta_0}[\psi_{\theta}(\widetilde{X}_j,\widetilde{T}_j)] = \alpha_{\theta_0} E_{\theta_0}(X_i) - \frac{\alpha_{\theta_0}}{\theta} + \frac{\alpha_{\theta_0}\dot{\alpha}_{\theta}}{\alpha_{\theta}} =:\Psi(\theta)$, 
		\item[(iv)] $E_{\theta_0}[\psi_{\theta_0}(\widetilde{X}_j,\widetilde{T}_j)]=\Psi(\theta_0)=0$ and
		\item[(v)] $\Psi_n(\hat{\theta}) \stackrel{p}{\to} 0$.
	\end{itemize}	
\end{lemma}

As a comparison, we consider the na\"ive approach to assume already for the observed data, $X_1, \ldots, X_m \stackrel{iid} \sim Exp(\theta_0)$. This is even more tempting, as the necessity of a population definition seems to be redundant.
Theoretically, under srs-assumption, the derivative of the log-likelihood  - multiplied by minus one - has summands
\begin{equation} \label{srspsi}
	\psi^{srs}_{\theta}(x_i)=x_i - \frac{1}{\theta},
\end{equation} 
being similar to the first two summands of \eqref{defpsi} if $i_j=1$. 
An interpretation of (ii) in Lemma \ref{allebeding} is now the srs-design as the limit, in the sense that, if $i_j=1$, it is, $\lim_{s \to \infty} \lim_{G \to \infty} \dot{\psi}_{\theta}(\tilde{x}_j,\tilde{t}_j) = \dot{\psi}^{srs}_{\theta}(x_i)$. Condition {\em (v)} is a tribute to boundary maxima,   $\Psi_n(\theta)$ has no zero in $\Theta$ in case of a too high or too low ``location'' of $\Psi_n$, in combination with a too small amplitude over the parameter space, meaning $\Psi_n(1/\varepsilon) - \Psi_n(\varepsilon)$.
As $\varepsilon$ can be chosen arbitrarily small, the amplitude depends on the limiting behaviour of $\Psi_n$ towards the boundaries of $\mathbb{R}^+$, on the left for $\theta \ssearrow 0$ and on the right for $\theta \to \infty$. Towards the left border, consider Taylor expansions for the numerator and denominator of $\psi_{\theta}(\tilde x_j, \tilde t_j)/i_j -  \tilde x_j$ to show that the first two derivatives, using l'H\^{o}spital's rule for $\theta \ssearrow 0$, are zero, but the third is not. The resulting {\em finite}  limit is
\begin{equation*} 
	\lim_{\theta \ssearrow 0} n \Psi_n(\theta)=M\frac{s + G}{2} - \sum_{i=1}^M X_i. 
\end{equation*}
Following up, note that 
\begin{equation} \label{boundmin2}
	\lim_{n \to \infty}\lim_{\theta \ssearrow 0} \Psi_n(\theta)= \alpha_{\theta_0} \left[\frac{s+G}{2} - E_{\theta_0}(X_i)\right]
\end{equation}
(see Definition \ref{defnnast} and Proof to Lemma \ref{allebeding}(iii)). Note further $\lim_{s \ssearrow 0}\alpha_{\theta_0} E_{\theta_0}(X_i)=0$, from Corollary \ref{EXi}(iii), and $\lim_{s \ssearrow 0}\alpha_{\theta_0}=0$ (see \eqref{alpha0}).   

Compare with  $\lim_{\theta \ssearrow 0}\psi^{srs}_{\theta}(x_i)=-\infty$, to see that the reduced amplitude implies less information for truncation, due to the obviously reduced slope also at $\theta_0$. 

By contrast, on the right border, the limiting behaviour for $\theta \to \infty$ is not affected by the change in design. To see when $\psi_{1/\varepsilon}(\tilde x_j,\tilde t_j)>0$, note that $\lim_{\theta \to \infty}\psi_{\theta}(\tilde x_j,\tilde t_j)/i_j -  \tilde x_j=0$, using l'H\^{o}spital's rule once. For the srs-design, it is the same and {\em finite}, showing that a boundary maximum can occur when the observed durations are small, i.e. when $\theta_0$ is large (compared to $n$).  We will continue the comparison of designs in Monte Carlo simulation and applications of Sections \ref{mcs} and \ref{appls}. 

\begin{theorem} \label{cons} Under assumptions \ref{A1:Compact}-\ref{A4:trunc} and for $\theta_0 \in  ]\varepsilon,1/\varepsilon[$ holds $\hat{\theta}  \stackrel{p}{\to} \theta_0$.
\end{theorem}

\begin{proof}
	Apply Lemma 5.10 in \cite{vaart1998}.  $]\varepsilon,1/\varepsilon[$ is a subset of the real line, $\Psi_n$ is a random function and $\Psi$ a fixed, both in $\theta$. It is $\Psi_n(\theta) \stackrel{p}{\to} \Psi(\theta)$ for every $\theta$, roughly speaking due to Lemma \ref{allebeding}(iii) and the LLN. Specifically, the Poisson property for $M$ results in $M/n \stackrel{p}{\to} \alpha_{\theta_0}$. Furthermore,   
	$\frac{1}{n} \sum_{j=1}^n I_j \widetilde{X}_j=\frac{1}{n} \sum_{i=1}^M X_i  \stackrel{p}{\to} \alpha_{\theta_0} E_{\theta_0}(X_i)$ is a consequence of $M \sim Poi(n \alpha_{\theta_0})$. Together with $E_{\theta_0}(M)= Var_{\theta_0}(M) = n \alpha_{\theta_0}$ one has 
	\begin{eqnarray*}
		Var_{\theta_0}\left(\frac{1}{n} \sum_{i=1}^M X_i \right) & = & E_{\theta_0}\left[Var_{\theta_0}\left(\frac{1}{n} \sum_{i=1}^M X_i |M\right)\right] \\
		& + & Var_{\theta_0}\left[E_{\theta_0}\left(\frac{1}{n} \sum_{i=1}^M X_i |M\right)\right] \\
		& = & \frac{1}{n^2} E_{\theta_0}[M Var_{\theta_0}(X_i)] + \frac{1}{n^2} Var_{\theta_0}[M E_{\theta_0}(X_i)] \\
		& = &  \frac{1}{n} \alpha_{\theta_0}  Var_{\theta_0}(X_i) + \frac{1}{n} [E_{\theta_0}(X_i)]^2 \alpha_{\theta_0} \quad \stackrel{n \to \infty}{\longrightarrow} 0,\\
	\end{eqnarray*}
	as $E_{\theta_0}(X_i)$  and $Var_{\theta_0}(X_i)$ are finite by Corollary \ref{EXi}(iii+iv). Convergence follows in squared mean, and hence in probability. 
	
	For the next condition in Lemma 5.10, we need a short discussion about maxima at the boundary of $\Theta$ for some -- typically small -- $n$. In these situations, there is no zero to $\Psi_n(\theta)$. We will demonstrate that, using the boundary in these situations, the MLE is a ``near zero''. That is, $\Psi_n(\theta)$ is non-decreasing due to Lemma \ref{allebeding}(ii) and Lemma \ref{allebeding}(v) holds.
	Furthermore,  $\Psi(\theta)$ is obviously differentiable and $\dot{\Psi}(\theta_0)>0$ with the same argument as for $\dot{\psi}_{\theta}$ in Lemma \ref{allebeding}(ii) for $(\tilde x_j, \tilde t_j)' \in D$, such that $\Psi(\theta_0- \eta) < 0 < \Psi(\theta_0+ \eta)$ for every $\eta>0$ when $\Psi(\theta_0)=0$, which holds due to Lemma \ref{allebeding}(iv).\qed    
\end{proof}

Although being the MLE, we cannot study asymptotic normality with general results from maximum likelihood theory. This would only be possible if we had considered an estimator for the pair $(n,\theta_0)$. Nonetheless, $\hat{\theta}$ is an M-estimator.

The main idea is to use the smoothness of $\Psi_n(\theta)$ and apply a quadratic Taylor expansion of $\Psi_n$ around $\theta_0$ and evaluated at $\hat{\theta}$, resulting in \cite[see][Equation (5.18)]{vaart1998}
\[
\sqrt{n} (\hat{\theta} - \theta_0) = \frac{- \sqrt{n} \Psi_n(\theta_0)}{\dot{\Psi}_n(\theta_0) + \frac{1}{2}(\hat{\theta} - \theta_0)\ddot{\Psi}_n(\tilde \theta)},
\]
with $\tilde \theta$ between $\hat{\theta}$ and $\theta_0$.  
We will need:
\begin{equation} \label{psi2}
	\begin{split}
		\psi_{\theta}^2(\tilde{x}_j,\tilde{t}_j) = & i_j \left( \frac{1}{\theta^2}  + \tilde x_j^2 + \frac{\dot{\alpha}_{\theta}^2}{\alpha_{\theta}^2} - \frac{2 \tilde x_j}{\theta} - \frac{2  \dot{\alpha}_{\theta}}{\theta \alpha_{\theta}} +  \frac{2 \tilde x_j  \dot{\alpha}_{\theta}}{\alpha_{\theta}} \right)\\
		\ddot{\psi}_{\theta}(\tilde{x}_j,\tilde{t}_j)  = & i_j \left( \frac{\dddot{\alpha}_{\theta}\alpha_{\theta} -   \dot{\alpha}_{\theta}\ddot{\alpha}_{\theta}}{\alpha^2_{\theta}} - \frac{2  \dot{\alpha}_{\theta} \ddot{\alpha}_{\theta} \alpha^2_{\theta} -  2\dot{\alpha}_{\theta}^3\alpha_{\theta}}{\alpha^4_{\theta}} -\frac{1}{2\theta^3} 
		\right) 
	\end{split}
\end{equation}

\begin{lemma}\label{restfnorm} It is $E_{\theta_0}[\psi^2_{\theta_0}(\widetilde{X}_j,\widetilde{T}_j)] < \infty$ and  $\ddot{\psi}_{\theta}(\tilde{x}_j,\tilde{t}_j) \le \ddot{\psi}(\tilde{x}_j,\tilde{t}_j)$ for all $\theta$ and the subsequent bound integrable. 
\end{lemma}

\begin{proof}
	For the first half: It is $I_j \widetilde{X}_j^2 \le (G+s)^2 \Rightarrow E_{\theta_0}(I_j \widetilde{X}_j^2) \le \alpha_{\theta_0} (G+s)^2$, $I_j \widetilde{X}_j \ge 0 \Rightarrow E_{\theta_0}(I_j \widetilde{X}_j) \ge 0$ and $I_j \widetilde{X}_j \le (G+s) \Rightarrow E_{\theta_0}(I_j \widetilde{X}_j) \le \alpha_{\theta_0} (G+s)$, so that 
	\begin{equation*} 
		\psi_{\theta_0}^2(\widetilde{X}_j,\widetilde{T}_j) \le  \frac{\alpha_{\theta_0}}{\theta_0}  + \alpha_{\theta_0} (G+s)^2 + \frac{( \dot{\alpha}_{\theta_0})^2}{\alpha_{\theta_0}}
		- \frac{2 \dot{\alpha}_{\theta_0}}{\theta_0} +  2 (G+s)  \dot{\alpha}_{\theta_0}
	\end{equation*}
	which is  finite due to $\theta_0 \in \Theta$, the finiteness and positivity of $\alpha_{\theta_0}$ from \eqref{alpha0} and the finiteness of $\dot{\alpha}_{\theta_0}$ from Corollary \ref{delalpha}(i). 
	For the second half: In \eqref{psi2}, we can replace the denominators by their (due to the arguments after \eqref{alpha0}) positive minima. Then, all numerators are continuous functions on compact $\Theta$ hence with finite maxima, that we may insert. So that   
	$\ddot{\psi}_{\theta}(\tilde{x}_j,\tilde{t}_j) \le i_j C=:\ddot{\psi}(\tilde{x}_j,\tilde{t}_j)$ (with $C < \infty$) having finite integral $C \alpha_{\theta_0}$. 
	\qed
\end{proof}

\begin{theorem}\label{asnorm} Let be $\theta_0 \in ]\varepsilon, 1/\varepsilon[$ then,  under assumptions \ref{A1:Compact}-\ref{A4:trunc}, holds
	$\sqrt{n} (\hat{\theta} - \theta_0) \stackrel{d}{\to} N(0, \sigma^2)$ with $	\sigma^2 :=  E_{\theta_0}(\psi_{\theta_0}^2(\widetilde{X}_j,\widetilde{T}_j))/	[E_{\theta_0}(\dot{\psi}_{\theta_0}(\widetilde{X}_j,\widetilde{T}_j))]^2$ (see definitions \eqref{psiklein0} and \eqref{psi2}).
\end{theorem}

\begin{proof}
	Use the classical assumptions of Fisher \cite[here in the formulation from][Theorem 5.41]{vaart1998}. The main assumption of consistency is Theorem \ref{cons}. Now $\psi_{\theta}(\tilde x_j, \tilde t_j)$ is twice continuously differentiable in $\theta$ for every $(\tilde x_j,\tilde t_j)$, due to Lemma \ref{allebeding}(i). $E_{\theta_0}[\psi_{\theta_0}(\widetilde{X}_j, \widetilde{T}_j)]=0$ due to Lemma \ref{allebeding}(iv) with $E_{\theta_0}[\psi^2_{\theta_0}(\widetilde{X}_j, \widetilde{T}_j)] < \infty$ due to Lemma \ref{restfnorm}. The existence of $E_{\theta_0}[\dot{\psi}_{\theta_0}(\widetilde{X}_j,\widetilde{T}_j)]$ follows from \eqref{defpsi} and positivity from Lemma \ref{allebeding}(ii) combined with $E_{\theta_0}(I_j)=\alpha_{\theta_0}>0$. Dominance of the second derivative by a fixed integrable function around $\theta_0$ is due to Lemma \ref{restfnorm}. 
	\qed
\end{proof}

For the estimation of the standard error (SE) from Theorem \ref{asnorm}, we replace expectations by averages over the latent sample,
\begin{equation} \label{seest}
	\frac{\hat{\sigma}}{\sqrt{n}} :=  \frac{\frac{1}{\sqrt{n}} \sqrt{\sum_{j=1}^n \psi_{\theta_0}^2(\tilde{x}_j,\tilde{t}_j)}}	
	{\frac{\sqrt{n}}{n} \sum_{j=1}^n \dot{\psi}_{\theta_0}(\tilde{x}_j,\tilde{t}_j)}, 
\end{equation}	
being observable, because indicators reduce sums up to $m$.

\section{Likelihood Approximation} \label{theory}

In order to give a precise version and derivation of the likehood \eqref{likeunsauber}, we now describe the truncated sample as stochastic process as in \cite{kalblawl1989}, especially as truncated empirical process, which in turn is approximated by a mixed empirical process. For the mixed process, deriving the likelihood is relatively simple. 

Denote by $\epsilon_a$ the Dirac measure concentrated at point $a \in S$. Define the  
point measure  $\mu := \sum_{j=1}^n \epsilon_{(\widetilde{x}_j,\widetilde{t}_j)'}$, 	$\mu : \mathcal{B} \mapsto \bar{\NN}_0$, and the space of point measure on $\mathcal{B}$ (with fixed $n$) by $\mathbb{M}$. By inserting random variables, it becomes an empirical process $N_n:=  \sum_{j=1}^n \epsilon_{(\widetilde{X}_j,\widetilde{T}_j)'(\omega)}$ ($\Omega  \mapsto \mathbb{M}$),
measurable w.r.t. $\sigma$-algebras from $\mathcal{A}$ to $\mathcal{M}$, the $\sigma$-algebra for $\mathbb{M}$.
The data is now  the truncated empirical  process (for an illustration, see Figure \ref{dtrunc}(left))
\begin{align*}
N_{n,D}(\cdot):= N_n(\cdot \cap D)=\sum_{j=1}^n \epsilon_{\binom{\widetilde{X}_j}{\widetilde{T}_j}} \left(\cdot \cap D \right),
\end{align*}
for which we write $X_1, \ldots, X_m$ in all but this section.
The size of the truncated sample is $N_{n,D}(S)$, for which we write $M$ - and realised $m$ - in all but this section, and is hence random and dependent on the sample size $n$. 

\begin{figure}[ht!] 
	\begin{center}
		
		\resizebox{0.50\textwidth}{!}{
			\begin{tikzpicture}

\begin{axis}[
view={15}{40},
scale only axis,
xmin=0,xmax=2.5,
ymin=0,ymax=2.5,
zmin=0,zmax=4.5,
axis lines=center,
mark=black,
xlabel= {\small $x$},
ylabel= {\small $t$},
zlabel= \empty,
xtick={1,2},
ytick={1,2},
ztick={0,1,2,3},
xtick align=inside,
ytick align=inside,
ztick align=inside,
clip=false
]

\addplot3[black, mark=,dashdotted] coordinates {
(0,0,0)
(2.5,2.5,0)

(1.25,0,0)
(2.5,1.25,0)
} ;

\addplot3[black, mark=*, mark size=1.5,only marks ] coordinates {
(1,0.2,0)
(2,1,0)
(1.5,1.2,0)
(2,0.3,0)
(0.5,1.2,0)

} ;

\addplot3[black, mark=,thick ] coordinates {
(1,0.2,0)
(1,0.2,1)
(1,2,1)

(1,0.2,1)
(2.5,0.2,1)

(1.5,1.2,1)
(1.5,1.2,2)
(1.5,2,2)

(1.5,1.2,2)
(2,1.2,2)

(2,1,1)
(2,1,2)
(2,1.2,2)

(2,1,2)
(2.5,1,2)

(2,1.2,2)
(2,1.2,3)
(2,2,3)

(2,1.2,3)
(2.5,1.2,3)
} ;

\addplot3[black, mark=,thin ] coordinates {
(1,0.2,0)
(2.25,0.2,0)

(1.5,1.2,1)
(1.95,1.2,1)

(2,1,1)
(2.5,1,1)

(2,1.2,2)
(2.5,1.2,2)
} ;

\addplot3[black, mark=,dotted] coordinates {
(2,1,0)
(2,1,2)

(1.5,1.2,0)
(1.5,1.2,2)

} ;
\node at (0,0,5.0) {\small $N_{n,D}([0,x] \times [0,t])$ };

\node at (1.3,-0.3,0) {\small $\D$};
\node at (2.3,1.8,0) {\normalsize $D$};
\end{axis}
\end{tikzpicture} 
		}
		\includegraphics[scale=0.28,angle=0]{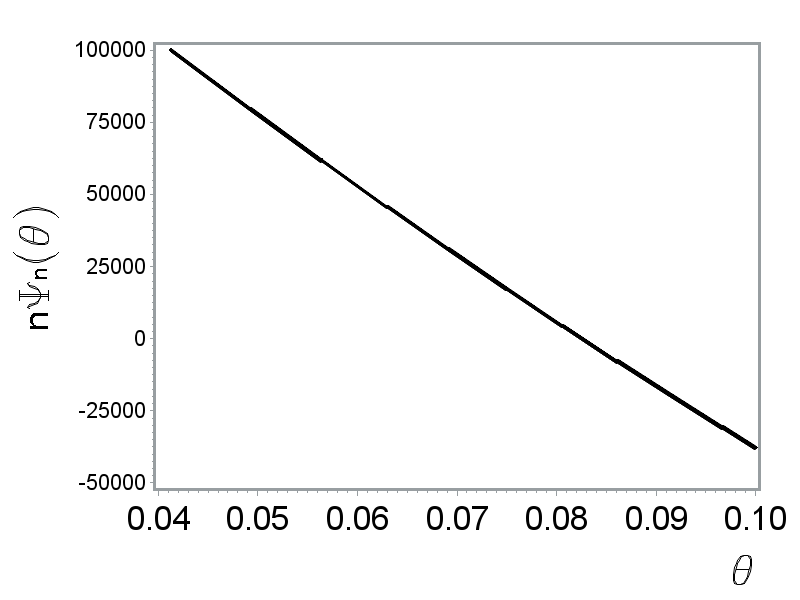}
	\end{center}
	\caption{Left: Realisation of $N_{n,D}$ on sequences of rectangles $[0,x] \times [0,t]$, as a function of the upper right corner $(x,t)'$. Dots mark $(\tilde{x}_j,\tilde{t}_j)'$. Right (for Section \ref{appls}): Criterion function \eqref{z_estimator} (times $n$) for Application ``insolvency''}  \label{dtrunc}
\end{figure}
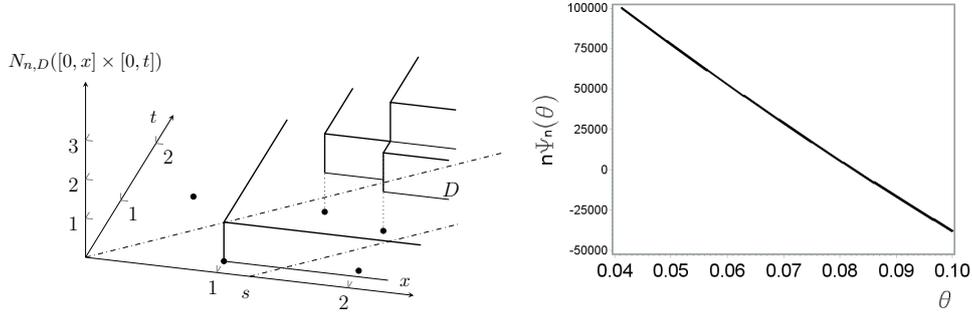

In order to parametrize the data, i.e. the truncated empirical process, we write its intensity measure (only needed for sets $[0,x] \times [0,t]$) as
\begin{equation} \label{IMtrunc}
\begin{split}	
\nu_{N_{n,D}}([0,x] \times [0,t]) := & E_{\theta_0}[N_{n,D}([0,x] \times [0,t])] \\
 = & n P \{(\widetilde{X}_j,\widetilde{T}_j)' \in [0,x] \times [0,t] \cap D\} \\
= & n \alpha_{\theta_0} F^{X,T}(x,t),
\end{split}
\end{equation}
due to Lemma \ref{lemma_imeasure}. To see that, note that 
\[
  \alpha_{\theta_0} F^{X,T}(x,t) =  \mathcal{L}(\widetilde{X}_j,\widetilde{T}_j)(E_1) 
=    F_E(x/\theta_0) F^{\widetilde{T}}(t) - \mathcal{L}(\widetilde{X}_j,\widetilde{T}_j)(E_2 \cup \E_3).
\]
Here, and in the following, the measure in the co-domain of a random variable is denoted $\mathcal{L}$, e.g. $\mathcal{L}(\widetilde{X}_j,\widetilde{T}_j)$. 
Note also that, $\nu_{N_{n,D}}$ evaluated at $S$, is $n \alpha_{\theta_0}$. One can show that $N_{n,D}$ is equal in distribution to a Binomial-mixing empirical process. However, as our data in the applications (Section \ref{appls}) will be relatively few, because $s$ is relatively small, we will see shortly that it is enough to approximate the data with a Poisson-mixing empirical process.

\begin{definition} \label{defnnast}
Assume \ref{A1:Compact}-\ref{A4:trunc} and  let $Z$ be Poisson-distributed with parameter $n \alpha_{\theta_0}$ and independent thereof $(X_i,T_i)'$ of Definition \ref{defxt}:    
	\[
	N^{\ast}_n:= \sum_{i=1}^{Z} \epsilon_{\binom{X_i}{T_i}}
	\]
\end{definition}
Due to $\nu_{n,D}(S)=n \alpha_{\theta_0}<\infty$ and  $\mathcal{L}[(X_i,T_i)']=\nu_{n,D}/(n \alpha_{\theta_0})$ (by \eqref{IMtrunc}) now  $N^{\ast}_n$ is a Poisson process  with an intensity measure \cite[see][Theorem 1.2.1(i)]{reiss1993}
\begin{equation} \label{nuast}
\nu_n^*=\nu_{n,D} \quad \text{and} \quad N^{\ast}_n(S)=Z.
\end{equation}
The latter is generally true for Poisson processes, (realized or not), so that $Z$ is also observed. 

The parallelogram  $D$ is ``small'' (in terms of $\mathcal{L}(\widetilde{X}_j,\widetilde{T}_j)$) relative to $S$, as long as the observation interval width $s$ is relatively small compared to the width $G$ of the population (and the typically long expected durations). Hence, $N^{\ast}_n$ is ``close'' to $N_{n,D}$ in Hellinger distance \cite[see e.g.][Approximation Theorem 1.4.2]{reiss1993}. We will now derive the likelihood for $N^{\ast}_n$.

The likelihood is the density of $N^{\ast}_n$, evaluated at the realisation, denoted as $n^{\ast}_n$, i.e. with inserted $z$ and $(x_i,t_i)'$'s. The density of $N^{\ast}_n$ has as its domain, the co-domain of $N^{\ast}_n$, $\mathbb{M}$, so that the density of $N^{\ast}_n$ is a function of the point measure $\mu$. Furthermore, a Radon-Nikodym density requires a dominating measure and we use the density of another Poisson process. We chose the 2-dim homogeneous Poisson process on $[0,A]^2$. 

\begin{definition}	\label{defn0}
Let $A \in \mathbb{N}$ be a number larger than the support of $X_i$ or $T_i$, e.g. the next natural number larger then $G+s$ (see Definition \ref{defxt}).  
Let $N_0$ be a Poisson process with $Z_0 \sim Poi_{A^2}$ and independently thereof $(X^0_i,T^0_i)' \sim Uni([0,A]^2)$ $i=1,2,3, \ldots$.
\end{definition}

Note that $N_0$ has a (finite) intensity measure, where $\lambda_{[0,A]^2}$ denotes the Lebegues measure restricted to $[0,A]^2$, \cite[see][Theorem 1.2.1.(i)]{reiss1993}
\begin{equation} \label{nu0}
\nu_0:= A^2 \lambda_{[0,A]^2}, \quad \text{including} \quad \nu_0(S)=A^2 \; \text{(see \eqref{nuast}(right))}.
\end{equation} 
The latter is different from a geometrically intuitive volume $A^4$. $\mathcal{L}(N_0)$ will now serve as the dominating measure in order to derive the Radon-Nikodym density of $\mathcal{L}(N^{\ast}_n)$. But for that we will need the Radon-Nikodym density of $\nu_{n,D}$ w.r.t. $\nu_0$, so that \cite[see][Formula (16.11)]{billingsley2012} one searches $h_{\theta_0}:S \rightarrow \mathbb{R}_0^+$
with $\forall B \in \mathcal{B}$ it is
\begin{equation} \label{defh}
\nu_{n,D}(B)= \int_B h_{\theta_0} \, d \nu_0.
\end{equation}
 For $B=[0,x]\times [0,t]$ and $x\le A, t \le A$ due to Fubini's theorem, with $\lambda$ as the univariate Lebesgues measure, due to the differentiability,  
\begin{eqnarray}\label{lemh}
\nu_{n,D}([0,x]\times [0,t]) & = & 
A^2 \int_0^x \int_0^t h_{\theta_0}(a_1,a_2) \lambda(d a_2) \lambda(d a_1) \nonumber \\
\Rightarrow h_{\theta_0}(x,t) & = & \frac{1}{A^2} \frac{\partial^2}{\partial x \partial t} \nu_{N_{n,D}}([0,x] \times [0,t])  \nonumber\\
& = & \frac{1}{A^2} \frac{\partial^2}{\partial x \partial t} n \alpha_{\theta_0} F^{X,T}(x,t) =\frac{n \theta_0}{G \, A^2} e^{- \theta_0 x},
\end{eqnarray}
where \eqref{IMtrunc} is used for the third equality, and Lemma \ref{lemma_imeasure}  for the forth together with 
$\frac{\partial^2}{\partial x \partial t} R(x,t) = 0$ from Lemma \ref{lemma_imeasure}. 
Of course, for $(x,t)' \not\in D$ is $h_{\theta_0}(x,t) = 0$.

\begin{theorem} \label{theolike}
	For Assumptions \ref{A1:Compact}-\ref{A4:trunc} and $\alpha_{\theta_0}$ from \eqref{alpha0}, the model $N^{\ast}_n$ of Definition \ref{defnnast}, has likelihood w.r.t. to $\mathcal{L}(N_0)$ from Definition \ref{defn0}:
	\begin{equation} \label{likelihood}
	\ell(n^{\ast}_n;\theta_0,n)= \frac{n^{n^{\ast}_n(S)} \theta_0^{n^{\ast}_n(S)}}{G^{n^{\ast}_n(S)} A^{2 n^{\ast}_n(S)}} \exp\left(-\theta_0 \sum_{i=1}^{n^{\ast}_n(S)} x_i\right) \exp (A^2 - n \alpha_{\theta_0})
	\end{equation}
\end{theorem}
The proof is in Appendix \ref{prooftheo1}. The main idea is to decompose the density of the data, i.e. of $\mathcal{L}(N_n^*)$, into the product of the density, conditional on $N_n^*(S)$, multiplied by the probability mass distribution of the Poisson distributed $N_n^*(S)$. The later results in the very last factor of \eqref{likelihood} to include an exponential function in $n \alpha_{\theta_0}$. Note that by Fisher-Neyman factorization $(N^{\ast}_n(S), \sum_{i=1}^{N^{\ast}_n(S)} X_i)$ is a sufficient statistic.

We maximise the likelihood as a function in its second argument, the ``generic'' parameter $\theta$, being already the notation  in \eqref{scoretheta}. For a thorough discussion about the parameter notation, we refer the reader to the maximum likelihood estimator as posterior mode in a Bayesian analysis with uniform prior \cite[see e.g.][Sect. 2.3]{Robert2001}. 
Finally note that, after taking logarithm, the derivatives w.r.t to $\theta$  and $n$ of \eqref{likelihood} are equal to that of its intuitive counterpart \eqref{likeunsauber} with $n^{\ast}_n(S)$ replaced by $m$ (see \eqref{scoretheta}).

\section{Monte Carlo Simulations} \label{mcs}

 Our aim in this section is twofold, first we illustrate the vanishing bias, i.e. consistency, stated theoretically by Theorem \ref{cons}. Second, the notion of a ``bias'', referring to {\em one} model so far, can be extended to the ``selection bias'' comparing two models. We will assess such design-effect compared to the srs-design as motivated theoretically after Lemma \ref{allebeding}.  
 
 We simulate $n \in \{10^p, p=3, \ldots, 6 \}$ durations $\widetilde{X}_j$ from Assumption \ref{A2:SquareInt} with $\theta_0 \in \{0.005, 0.01, 0.05, 0.1 \}$ according to \ref{A1:Compact} and further $\widetilde{T}_j$ according to \ref{A2:SquareInt} with $G \in \{24, 48 \}$, and we obey \ref{A3:Ind}. We then retained $m$ of the $\tilde{x}_j$, that fulfil \ref{A4:trunc} with $s \in \{2,3,48\}$. 
 We calculate for the data set $v$ the MLE $\hat{\theta}^{(v)}$ as zero of \eqref{z_estimator} by means of a standard algorithm. Boundary maxima do not occur because \eqref{boundmin2} is markedly negative for all simulation scenarios. 
 
\begin{sidewaystable}  \caption{Simulation averages of Bias, estimated asymptotic variance $\hat{\sigma}^2$ and $\widehat{VIF}$ / Simulated $Var(\hat{\theta})$ (times $n$)}
	{\footnotesize 
		\setlength\extrarowheight{-15pt}
		\begin{center}
			\begin{tabular}{clllll|cllll} \hline\noalign{\smallskip}
				& & \multicolumn{4}{c}{$G=24, s=3$} & \multicolumn{5}{c}{$G=24, s=48$} \\ \hline \hline \noalign{\smallskip}
				$\underline{\theta_0}$ 	&     &     $n=1000$  &        $n=10,000$ &   $n=100,000$ & $n=10^6$ &  &  $n=1000$  &        $n=10,000$ &   $n=100,000$ & $n=10^6$ \\ \hline \noalign{\smallskip}
				\underline{0.005} &	Bias &   0.0132 &   0.0023 &   -0.000035 &   -0.000007 &               & 0.00049 & -0.00002 & -0.000007 & -0.000005 \\
			 &	$\hat{\sigma}^2$ &   0.862 &  0.96 & 1.307 &   1.463 &   	       & 0.018 & 0.0207 & 0.0208 & 0.0208 \\ 
				$\boldsymbol{\alpha}_{\boldsymbol{\theta}_{\mathbf{0}}}=0.014$  &	$n Var(\hat{\theta})$ &   0.69 &  0.705 &  1.23 &  1.498 &   $\boldsymbol{\alpha}_{\boldsymbol{\theta}_{\mathbf{0}}}= 0.201$      & 0.0170 & 0.0205 & 0.02105 & 0.0204 \\
				 &	$\widehat{VIF}$ &  1.6 &  2.25 &  3.21 &   3.61  &           & 4.3 & 5.1 & 5.1 &  5.1 \\ \hline \noalign{\smallskip}
				\underline{0.01} &	Bias &  0.007 & 0.00004 & -0.00002 & 0.00006 &         &  -0.00005 & 0.00003 & -0.000016 & 0.000003 \\
				 &	$\hat{\sigma}^2$ &  0.53 & 0.688& 0.78 & 0.78 &      	             & 0.0124 & 0.0124 & 0.0124 & 0.0124\\ 
				$\boldsymbol{\alpha}_{\boldsymbol{\theta}_{\mathbf{0}}}=0.026$  	  &	$n Var(\hat{\theta})$ &  0.40 & 0.606 & 0.855 & 0.745 &    $\boldsymbol{\alpha}_{\boldsymbol{\theta}_{\mathbf{0}}}=0.339$          & 0.0121 & 0.0123 & 0.0121 & 0.0123 \\	
					 &		$\widehat{VIF}$ &  2.02 & 3.02 &   3.48 &   3.48 &              &  4.74 &   4.75 &  4.75 &  4.75 \\ \hline \noalign{\smallskip}
				\underline{0.05} &	Bias & 0.0009 & 0.0003 & 0.0000009 & 0.00001 &           & -0.0007 & -0.00002 & 0.0000098 & 0.000002 \\
					 &	$\hat{\sigma}^2$  &  0.277 & 0.271 & 0.271 & 0.271 &    	        &  0.0098 & 0.00984 & 0.0098 & 0.0098 \\ 
				$\boldsymbol{\alpha}_{\boldsymbol{\theta}_{\mathbf{0}}}=0.081$	  &		$n Var(\hat{\theta})$ &  0.280 & 0.262 & 0.265 & 0.265 &  $\boldsymbol{\alpha}_{\boldsymbol{\theta}_{\mathbf{0}}}=0.530$    & 0.0103 & 0.0097 & 0.0096 & 0.0099 \\
				  &	$\widehat{VIF}$ &    2.71 &     2.71 &    2.72 &    2.72 &              & 3.23 &    3.22 &   3.22 &   3.22 \\	\hline \noalign{\smallskip}
				\underline{0.1} &	Bias  &   0.0013 & -0.00004 & -0.000008 & -0.00002 &        & 0.00014 & -0.000026 & 0.000018 & 0.000003 \\
					 &	$\hat{\sigma}^2$ &   0.276 & 0.271 & 0.271 & 0.271 &    	    &   0.0227 & 0.0226 & 0.0226 & 0.0226 \\ 
			$\boldsymbol{\alpha}_{\boldsymbol{\theta}_{\mathbf{0}}}=0.098$		  & 	$n Var(\hat{\theta})$ &  0.287 & 0.2708 & 0.258  & 0.276 &    $\boldsymbol{\alpha}_{\boldsymbol{\theta}_{\mathbf{0}}}=0.376$      & 0.0229 & 0.0225 & 0.0224 & 0.0224 \\
					  &	$\widehat{VIF}$  &  2.15 &  2.17 &   2.17 &    2.17 &         & 2.52 &   2.52 &   2.52 &   2.52 \\ \hline \hline \noalign{\smallskip}
				$\theta_0$	& & 	\multicolumn{4}{c}{$G=48, s=3$}	 & 	\multicolumn{5}{c}{$G=24, s=2$} \\ \hline \hline	\noalign{\smallskip}	
				\underline{0.005} &	Bias &  0.0065 & 0.00046 & -0.00004 & 0.00001 &           & 0.018 & 0.004 & 0.0003 & -0.00003 \\
				  &		$\hat{\sigma}^2$	&  0.267 & 0.313 & 0.391 & 0.393 &   	      & 1.27 & 1.38 & 1.845 & 2.21  \\ 
				$\boldsymbol{\alpha}_{\boldsymbol{\theta}_{\mathbf{0}}}=0.013$  &		$n Var(\hat{\theta})$	&	 0.23 & 0.26 & 0.397 & 0.347 &    $\boldsymbol{\alpha}_{\boldsymbol{\theta}_{\mathbf{0}}}=0.009$       & 1.21 & 1.14 & 1.81 & 2.38 \\
				 &		$\widehat{VIF}$ &  1.56 & 2.40 & 3.11 &  3.14 &         & 1.30 & 1.94 & 2.78 & 3.37  \\ \hline \noalign{\smallskip}		
				\underline{0.01}  &	Bias		& 0.0026 & -0.00003 & -0.00002 & 0.000016 &            & 0.010 & 0.001 & -0.00005 & 0.00009 \\
				 &	$\hat{\sigma}^2$ &  0.174 & 0.221 & 0.224 & 0.224 &     	        & 0.78 & 0.98 & 1.179 & 1.18  \\ 
				$\boldsymbol{\alpha}_{\boldsymbol{\theta}_{\mathbf{0}}}=0.023$ &	$n Var(\hat{\theta})$ &	 0.149 & 0.213 & 0.226 & 0.235 &     $\boldsymbol{\alpha}_{\boldsymbol{\theta}_{\mathbf{0}}}=0.018$     & 0.62 & 0.84 & 1.268 & 1.12 \\
				 &	$\widehat{VIF}$ &  1.96 &  2.87 &  2.92 &  2.92 &             & 1.69 &  2.64 &   3.25 & 3.25  \\ \hline \noalign{\smallskip}
				\underline{0.05} &		Bias &	 0.0005 & 0.00009 & 0.00005 & -0.000002 &              & 0.001 & 0.0001 & 0.00003 & 0.000004 \\
				 &		$\hat{\sigma}^2$ &  0.131 & 0.129 & 0.128 & 0.128 &   	     & 0.41 & 0.40 & 0.40 & 0.40  \\ 
				$\boldsymbol{\alpha}_{\boldsymbol{\theta}_{\mathbf{0}}}=0.053$ &	$n Var(\hat{\theta})$	&	 0.132 & 0.119 & 0.125 & 0.147 &   $\boldsymbol{\alpha}_{\boldsymbol{\theta}_{\mathbf{0}}}=0.055$          & 0.43 & 0.39 & 0.396 & 0.39  \\
				  &	$\widehat{VIF}$ &    1.87 &      1.88 &    1.88 &   1.88 &           & 2.49 &   2.51 &   2.51 &   2.51 \\ \hline \noalign{\smallskip}		
				\underline{0.1} &	Bias  &	  0.0023 & -0.0001 & -0.00001 & -0.0000008 &             & 0.002 & -0.0002 & -0.00005 & -0.000026 \\
				&		$\hat{\sigma}^2$ &   0.238 & 0.229 & 0.229 &  0.229 &   	     &  0.40 & 0.39 & 0.39 & 0.39 \\
				$\boldsymbol{\alpha}_{\boldsymbol{\theta}_{\mathbf{0}}}=0.054$  & 	$n Var(\hat{\theta})$ &	 0.268 & 0.207 & 0.228 & 0.2296 &    $\boldsymbol{\alpha}_{\boldsymbol{\theta}_{\mathbf{0}}}=0.069$       & 0.44 & 0.41 & 0.37 & 0.41 \\
				 &		$\widehat{VIF}$ &   1.46 &     1.49 &   1.49 &   1.49 &      & 1.96 &     1.98 &     1.98 &    1.98  \\
				\hline \noalign{\smallskip}\hline\noalign{\smallskip}
		\end{tabular} \end{center}
	} \label{sims}
\end{sidewaystable}

In order to illustrate, first, consistency, assess the finite sample bias as an average over the $R=1000$ simulated $(\hat{\theta}^{(v)} - \theta_0)$. Table \ref{sims}(1$^{st}$ rows) lists the results, and it can be seen that the bias decreases to virtually zero. In order to show the decline in the mean squared error, consider the estimated standard error \eqref{seest} of $\hat{\theta}^{(v)}$. In Table \ref{sims}(2$^{nd}$ rows) averages over the $\left(\hat{\sigma}^{(v)}\right)^2$ seem to have a finite limit for increasing $n$. Hence, the standard error decreases of order $\sqrt{n}$. 

A by-product of the simulations is that they enable confirming the representation of $\sigma^2$ (in Theorem \ref{asnorm}). On the one hand,  $Var(\hat{\theta})$ can be approximated by $\frac{1}{R} \sum_{v}^R (\hat{\theta}^{(v)} - \theta_0)^2$, the simulated variance, i.e. $\sigma^2= Var(\sqrt{n} \hat{\theta})  =n Var(\hat{\theta})$ by $n$ times the simulated variance (Table \ref{sims}(3$^{rd}$ rows)). On the other hand, in a simulation, and not in an application, can $\sigma^2$ be estimated as $n$ times the square of \eqref{seest} (Table \ref{sims}(2$^{rd}$ rows)). Both quantifications become equal for large $n$.

The relation of the standard error with respect to $\alpha_{\theta_0}$ is also interesting. It decreases, obviously because $\alpha_{\theta_0}$ is linearly related to the size of the truncated sample  by  $m=n \alpha_{\theta_0}$ (see again \eqref{IMtrunc}).  The relation of $\alpha_{\theta_0}$ to $\theta_0$, $s$ and $G$ is already explained after its definition \eqref{alpha0} and its respective sensitivity is presented in Table \ref{sims}.
There is one exception; although $\alpha_{\theta_0}$ is decreasing in $G$, the simulated $\hat{\sigma}^2$ does not increase, but instead decreases for a given $n$ (Table \ref{sims}(left panels)). The reason can be suspected to be as in the srs-design, where the estimated standard error \eqref{estsesrsr} is not only increasing in $m$ of order $1/2$, but also decreases in  $\sum_{i=1}^m x _i$ of order $1$, the latter being much larger for a large $G$ (at given $m$).

Second, for the srs-design, applying \eqref{srspsi} results in an MLE $\hat{\theta}_{srs}= m/\sum_{i=1}^m x _i$ with standard error $\sigma_{srs}/\sqrt{m}= \theta_0/\sqrt{n^{\ast}(S)}$ (i.e. $\sigma_{srs}^2:=\theta_0^2$). The latter can be estimated by inserting $\hat{\theta}_{srs}$,
\begin{equation} \label{estsesrsr}
	\frac{\hat{\sigma}_{srs}}{m} = \frac{\sqrt{m}}{\sum_{i=1}^m x _i}.
\end{equation}
The factor for ``inflating'' the variance from Theorem \ref{asnorm}, denoted as Kish's design effect, is
\begin{equation} \label{vif}
	VIF:=  \frac{\sigma^2/n}{\sigma^2_{srs}/m}. 
\end{equation}
Illustrating the design effect with the $VIF$ is typical for the field of sampling techniques, especially in survey sampling. (By contrast, in the field of econometrics, variance inflation typically denotes the fact that standard errors increase for coefficients in a regression when accepting more covariates.) In the simulations, the $VIF$ remains overall at a quite moderate size, with a tendency to increase in $\alpha_{\theta_0}$.

  We will continue the comparison of designs in the applications of Section \ref{appls} where we will see a substantial variance inflation in all three applications.

 \section{Three Empirical Applications} \label{appls}
 
\subsection{Populations and Data}
 
 \subparagraph{Insolvency of Corporates founded 1990 -- 2013.}  The population of our first application are German companies founded after the last structural break in Germany, the re-unification, namely at the beginning of 1990. The first event is the foundation of the company, and the second considered event is the insolvency. We restrict attention to the $G=24$ years until the end of 2013, after which we started observing. Let $\widetilde{X}_j\sim Exp(\theta_0)$ denote the age-at-insolvency, and by $\widetilde{T}_i$ its age at the beginning of 2014. We assume a foundation to have taken place constantly (over those $G=24$ years), i.e. $\widetilde{T}_j \sim Uni[0,24]$. The German federal ministry of finance publishes the age of each insolvent debtor. We stop observing in 2016, i.e. $s=3$, after having collected, as a truncated sample $m=55,279$ companies.

\subparagraph{Divorce of Couples Married 1993 -- 2017.} In our next application, the German bureau of statistics reports divorces, with marriage lengths. Of marriages sealed between 1993 and 2017 in the German city of Rostock, $m=327$ marriages were divorced during 2018. Of these, 82 lasted less than 5 years, 112 lasted 6-10, 67 lasted 11-15, 40 lasted 16-20 and 26 held 21-25 years, i.e. $G=25$ and $s=1$. This small sample size example can help to understand dependence of the variance inflation to the data size.  
 
 \subparagraph{Dementia Onset of People Born 1900 -- 1954.}
  Our final application is dementia incidence in Germany for the birth cohorts 1900 until 1954. The first event is the 50$^{th}$ birthday of a person, between 1950 and 2004, i.e. we have $G=55$ .  An insurance company reported that between 2004 and 2013 ($s=9$), $m=35,929$ insurants has had a dementia incidence (the second event) \cite[for more information about the data see][]{weisetal2020}.

\subsection{Comparison of Estimation Results}

 The zero of \eqref{z_estimator}, i.e. the point estimate $\hat{\theta}$, is found graphically, for instance for the first application by Figure \ref{dtrunc}(right). For the estimated standard error see \eqref{seest}. All estimates are in Table \ref{alleschaetzer}, which also contains the estimates under srs-design \eqref{estsesrsr}.  
\begin{table}
	\caption{ML estimate $\hat{\theta}$ and estimate of standard error (SE) $\sigma/\sqrt{n}$ (see Theorem \ref{asnorm}) for applications, and comparison with srs-design}\label{alleschaetzer}  
	\begin{tabular}{llccc} \hline \hline
		&			& insolvency & divorce & dementia \\ \hline
		$G/s$ & in years & $24/3$ & $25/1$ & $55/9$ \\
		$m$ &  & $55,279$ & $327$ & $35,929$ \\
		$\sum_{i=1}^m x_i$ & in mio. years  & $0.54$ & $0.003$ & $1.1$ \\ 
		$\sum_{i=1}^m x_i^2$ & in mio. years$^2$ &  $2.5$ & $0.046$ & $36.3$ \\\hline 		
		truncation design & point estimate ($\hat{\theta}$)  & 0.08 & 0.066 & 0.0055 \\
		& $\widehat{SE}$: $\hat{\sigma}/\sqrt{n}$ & $0.00067$ & $0.0082$ & $0.0003$ \\ \hline 
		srs-design & point estimate ($\hat{\theta}_{srs}$) & 0.103 & 0.101 & 0.033 \\
		& $\widehat{SE}$: $\hat{\sigma}_{srs}/\sqrt{m}$ & $0.00044$ & $0.0056$  & $0.00017$  \\ \hline  \hline
	\end{tabular}
\end{table}

It is evident that ignoring truncation overestimates the hazard $\theta_0$ by, for example, 29\% in the insolvency application, and also causes negative selection of units in the others. We observe that the standard error is underestimated by about 35\% for all applications (equivalent to an on average $\widehat{VIF}=2,5$, as estimation of  \eqref{vif}), presumably through ignoring the stochastic dependence between units (and thus measurements) within the truncated sample. Also variance inflation almost seems not to depend on the sample size.

\section{Discussion} \label{secdiscussion}

The results are encouraging, as even after truncation, asymptotic normality holds, and standard errors do not increase too much. The considerable selection bias can be accounted for easily and identification of the parameters follows from standard results on the exponential family.

However, it is somewhat unfortunate that standard consistency proofs for the Exponential family fail, because compactness of the parameter space is violated, even when re-parametrising, due to the growing sample size being a parameter itself. And a temptation to withstand is to misinterpret the data as a simple random sample, only because statistical units are selected with equal probabilities (see \eqref{alpha0}). This is especially tempting, because if the truncated sample was simple, not knowing $n$ would be similar to not knowing the size of the population, requiring ``finite-population corrections'' only in the case of relatively many observations.      

In practice, the considerable effort to account for truncation can even be circumvented in rich data situations by adjusting the population definition to start at the observation interval, however thereby excluding {\em observable} units \cite[see e.g.][]{testing-ho:2006}. 

Of course more advanced sampling designs exist, such as endogenous sampling where units that have had a longer timeframe have a larger selection probability, in contrast to our model (sse \eqref{alpha0}). Also truncation is typically analysed with counting process theory, focusing more on the role of the filtration as an information model \cite[see e.g.][]{Andalt}. And with respect to robustness, the maximum likelihood method we use can be inferior to the method of moments \cite[see e.g.][]{weisradl2019,rowied2019}.  

Nonetheless, we believe that our approach still offers some advantages: As we (i) directly recognize the second measurement, the age when observation starts, as random, (ii) model the sample size as random and (iii) distinguish explicitly between indices in observed and unobserved sample. 

Two more minor points appear notable. First, the distance from the data to the mixed empirical process can be reduced to zero by changing from Poisson-mixing to Binomial-mixing, although little new insight can be expected, other than longer proofs. The same is true when proving the information equality for the standard error. And finally, one troublesome aspect should not be concealed. Compare the design effect with the theory of cluster samples where the $VIF$ increases in the cluster size linearly, for given intra-cluster correlation. Considering the time as a classifier, truncation seems to introduce a very small intra-temporal correlation, because the increase in the VIF is small. However, for very small sample sizes, the $VIF$ should then be even smaller. Non-linear behaviour of the dependence on the sample size is conceivable.

\textbf{Acknowledgment}:
The financial support from the Deutsche Forschungsgemeinschaft (DFG)
of R. Wei\ss bach is gratefully acknowledged (Grant WE 3573/3-1 ``Multi-state, multi-time, multi-level analysis of health-related demographic events: Statistical aspects and applications''). we thank W. Lohse, D. Ollrogge and G. Doblhammer for support in the data acquisition process. For the support with data we thank the AOK Research Institute (WIdO). The linguistic and idiomatic advice of Brian Bloch is also gratefully acknowledged.


\begin{appendix}

	\section{Proof of Lemma \ref{allebeding}} \label{prooflem2}

		For (i), note first that by \ref{A1:Compact}, $\theta_0 \in \Theta$, so is $\theta$: For $(\tilde{x}_j,\tilde{t}_j) \not\in D$, $\psi_{\theta}(\tilde{x}_j,\tilde{t}_j) \equiv 0$. Alternatively, due to Corollary \ref{delalpha}, first and second derivatives of $\alpha_{\theta}$ are continuous, and therefore, so will be the third. Also $\alpha_{\theta}$, being - along with $\theta$ - the only component of a denominator in the first or second derivative of $\psi_{\theta}$, is strictly positive due to the quotient rule. \\  
		For (ii): 
		For the equality, due to Corollary \ref{delalpha}, it is
		\begin{equation*} 
		\begin{split}
		\dot{\psi}_{\theta}(\tilde{x}_j,\tilde{t}_j) &=			
		i_j \left(\frac{1}{\theta^2} -\frac{\dot{\alpha}_{\theta}^2}{\alpha^2_{\theta}}   + \frac{\ddot{\alpha}_{\theta}}{\alpha_{\theta}}\right) \\		
		= &  \frac{i_j}{\theta^2} + i_j 
		\left(
		\frac{-s^2 e^{-\theta s}(1-e^{- \theta s}) - s^2 e^{-2 \theta s}}{(1- e^{- \theta s})^2} \right. \\
		& \left. + \frac{-G^2 e^{-G \theta}(1-e^{- G \, \theta}) - G^2 e^{-2 G \, \theta}}{(1- e^{- G \, \theta})^2} + \frac{1}{\theta^2} \right). \\
		\end{split}
		\end{equation*}
		For the positivity, we start to show that for $x>0$ or $y>0$
		\begin{eqnarray*}
		x e^{-x/2} & < & 1 - e^{-x}	 \Leftrightarrow 2 \frac{x}{2} e^{-x/2}  <  1 - e^{-2 \frac{x}{2}} \\
		2y e^{-y} & < & 1 - e^{-2y} \Leftrightarrow g(y):= 1 - e^{-2y} - 2y e^{-y} > 0
		\end{eqnarray*}
		Study its slope, $	g'(y) =  2 e^{-2y} - (2e^{-y} - 2y e^{-y}) =	2 e^{-2y} - 2e^{-y}  + 2y e^{-y}$, being equal to zero if and only if 
			\begin{eqnarray*}
  	 e^{-2y} - e^{-y}  + y e^{-y} = 0 	&	\Leftrightarrow &  	e^{-2y} =   (1-y) e^{-y} \\
				\Leftrightarrow   -2y = \log (1-y) -y 	&	\Leftrightarrow & e^{-y} = 1-y.
			\end{eqnarray*}
		The latter is only fulfilled for $y=0$, due to the known inequality $e^y > 1+y$ for $y \neq 0$, applied to $-y$. Now, $y=0$ is not in the domain and hence, $g$ does not change the sign of the slope. 
		It is $g(\log(2))=0.06$ and $g(1)= 0.13$, so that $g$ is increasing and positive, due to $\lim_{y \to 0} g(y)=0$. Now proceed to observe that from $x e^{-x/2} <  1 - e^{-x} \Rightarrow x^2 e^{-x}  <  (1- e^{-x})^2$ follows 		 
		\begin{equation*}
		\frac{s^2 e^{-\theta s}}{(1-e^{-\theta s})^2}  <  \frac{1}{\theta^2} 
		\end{equation*}
		and similarly for $G$ instead of $s$, both for $i_j=1$.

		For (iii): 
		\begin{eqnarray*}
			E_{\theta_0}[\psi_{\theta}(\widetilde{X}_j,\widetilde{T}_j)] & = & 	E_{\theta_0}[\Psi_n(\theta)] \\
			& = & 	 E_{\theta_0}\left(\frac{1}{n}\sum_{i=1}^M X_i  - \frac{M}{n \theta}  +   
			 \frac{M}{n \alpha_{\theta}}  \dot{\alpha}_{\theta}\right) \\
			& = & \frac{1}{n} E_{\theta_0}\left[E_{\theta_0}\left(\sum_{i=1}^M X_i| M\right)\right]   - \frac{ \alpha_{\theta_0}}{\theta}  +
			 \frac{ \alpha_{\theta_0} \dot{\alpha}_{\theta}}{\alpha_{\theta}} 
		\end{eqnarray*}	
		For (iv): For the first equality, note that $E_{\theta}[\psi_{\theta_0}(\widetilde{X}_j,\widetilde{T}_j)]=\Psi(\theta)$ due to (iii). Further, because of \eqref{alpha0}, Corollary \ref{EXi}(iii) and  Corollary \ref{delalpha}, we have:
		\begin{eqnarray*}
			\theta_0^2 \frac{\alpha_{\theta_0}}{\theta_0} & = & \theta_0^2 \left(\frac{1}{G  \theta_0^2} -  \frac{e^{-\theta_0 s}}{G  \theta_0^2} - \frac{e^{-\theta_0 \, G}}{G  \theta_0^2} + \frac{e^{-\theta_0(G+s)}}{G  \theta_0^2} \right) \\
			& = & \frac{1}{G} -  \frac{e^{-\theta_0 s}}{G} - \frac{e^{-\theta_0 \, G}}{G} + \frac{e^{-\theta_0(G+s)}}{G}
		\end{eqnarray*}
			\begin{eqnarray*}
			- \theta_0^2 \dot{\alpha}_{\theta_0} & = &  \theta_0^2 \left(-	\frac{s}{G \theta_0} e^{-\theta_0 s} + \frac{s}{G \theta_0} e^{-(G+s) \theta_0} - \frac{1}{\theta_0}e^{-G  \theta_0} +  \frac{1}{\theta_0} e^{-(G+s) \theta_0} \right. \\
			& & 
			\left. + \frac{1}{G \theta_0^2} 
			- \frac{1}{G \theta_0^2} e^{-\theta_0 s} - \frac{1}{G \theta_0^2} e^{-\theta_0 G} + \frac{1}{G \theta_0^2} e^{-(G+s) \theta_0} \right) \\
			& = &   -	\frac{s\theta_0}{G} e^{-\theta_0 s} + \frac{s \theta_0}{G} e^{-(G+s) \theta_0} - \theta_0 e^{-G  \theta_0} +  \theta_0 e^{-(G+s) \theta_0}  \\
			& & 
			+ \frac{1}{G} 
			- \frac{1}{G} e^{-\theta_0 s} - \frac{1}{G} e^{-\theta_0 G} + \frac{1}{G} e^{-(G+s) \theta_0}  \\
			& = & \left( -	\frac{s\theta_0}{G} - \frac{1}{G}\right) e^{-\theta_0 s} + \left(  - \theta_0 - \frac{1}{G} \right) e^{-\theta_0 G} \\
			& & + \left(\frac{s \theta_0}{G} +  \theta_0 + \frac{1}{G}\right) e^{-(G+s) \theta_0}  + \frac{1}{G} 
		\end{eqnarray*}
			\begin{eqnarray*}
			- \theta_0^2  \alpha_{\theta_0} E_{\theta_0}(X_i) & = & \theta_0^2 \left[ \left(\frac{s}{G\theta_0}+ \frac{2}{G\theta_0^2}\right) e^{-\theta_0 s} 
			- \left(- \frac{1}{\theta_0} -  \frac{2}{G \theta_0^2}\right) e^{-\theta_0 G} \right. \\
			& & \left. - \left( \frac{G+s}{G \theta_0} +  \frac{2}{G \theta_0^2} \right) e^{-\theta_0 (G+s)} - \frac{2}{G\theta_0^2} \right] \\
			& = &  \left(\frac{s \theta_0}{G}+ \frac{2}{G}\right) e^{-\theta_0 s} 
			+ \left( \theta_0 +  \frac{2}{G}\right) e^{-\theta_0 G} \\
			& &  + \left( - \frac{(G+s)\theta_0}{G} -  \frac{2}{G} \right) e^{-\theta_0 (G+s)} - \frac{2}{G} 
		\end{eqnarray*}
			The three terms add up to $-\theta_0^2 \Psi(\theta_0)$ of (iii) and adding the coefficients of $e^{-\theta_0 s}$, $e^{-\theta_0 \, G}$ and $e^{-\theta_0 (s+G)}$ (and the constants), we have $\theta_0^2 \Psi(\theta_0)=0$. 
		Finally, it is $\theta_0 \neq0$.
		
		For (v): The main idea of the proof is that in the event of a boundary minimum, the distance from $\Psi_n(\theta)$ to the $\theta$-axis is smaller than to $\Psi(\theta)$, and that it will converge to the latter. Hence, after surpassing the axis, there will be a zero and $\Psi_n(\hat{\theta})=0$.
		
		We need to show, stressing the dependence of $\hat{\theta}$ on $n$, that:
		\begin{eqnarray*}
			P\{ |\Psi_n(\hat{\theta}_n)| > \eta \} & \to & 0 \; \text{for} \; \eta > 0.
		\end{eqnarray*}	
	Denote the 'event' of a boundary minimum on the left side as  (recall the monotonicity of $\Psi_n(\theta)$ from (ii)), $A_n  :=  \{\hat{\theta}_n = \varepsilon \} =  \{ \Psi_n(\varepsilon) > 0\}$,
	and on the right as $B_n  :=  \{\hat{\theta}_n = 1/\varepsilon \}=\{ \Psi_n(1/\varepsilon) < 0\}$. Again due to the monotonicity of $\Psi_n(\theta)$, the events are mutually exclusive, $	A_n \cap B_n =  \emptyset $, with the consequence that $P(A_n \cup B_n)  =  	P(A_n) +	P(B_n)$.

	Recall that $\Psi(\theta_0)=0$ (from (iv)). Also it is $\dot{\Psi}(\theta)> 0$ with the same calculation as for $\dot{\psi}_{\theta}(\tilde x_j,\tilde t_j)$ in the equality of \eqref{psiklein0} (for $(\tilde x_j,\tilde t_j) \in D$) in (ii).
	Hence, for $\theta_0 \in ]\varepsilon, 1/\varepsilon[$, $\Psi(\theta)$ is 'away' from zero at the boundary, i.e. $-\Psi(\varepsilon) > 0$ and $\Psi(1/\varepsilon) > 0$
	Furthermore, in the event of $A_n$, the distance from $\Psi_n(\varepsilon)$ to the $\theta$-axis is smaller than to (the negative) $\Psi(\varepsilon)$: 
	\begin{equation} \label{wichtig}
	\Psi_n(\varepsilon)\le |\Psi_n(\varepsilon) - \Psi(\varepsilon)| 
	\end{equation}
Similarly, in the event of $B_n$, it is
		\begin{equation} \label{wichtig2}
- \Psi_n(1/\varepsilon) \le |\Psi_n(1/\varepsilon) - \Psi(1/\varepsilon)|
	\end{equation}
We have	$|\Psi_n(\hat{\theta}_n)| > \eta  \Leftrightarrow  \Psi_n(\hat{\theta}_n) \neq 0 
	  \Leftrightarrow  \hat{\theta}_n \in \{\varepsilon, 1/\varepsilon\}  \Leftrightarrow A_n \cup B_n	   \Leftrightarrow  
	  \{  \Psi_n(\varepsilon) > 0\} \cup  \{  \Psi_n(1/\varepsilon) < 0\}$ and hence
\begin{equation*}
\begin{split}
P\{|\Psi_n(\hat{\theta}_n)| > \eta	\} & =  P\{ \{ \Psi_n(\varepsilon) > 0\} \cup  \{  - \Psi_n(1/\varepsilon) > 0\}  \} \\
 = & P\{  \Psi_n(\varepsilon) > 0 \} + \{ - \Psi_n(1/\varepsilon) > 0\}   \\
 \le &  P\{  |\Psi_n(\varepsilon) - \Psi(\varepsilon)| >0 \} +   P\{  |\Psi_n(1/\varepsilon) - \Psi(1/\varepsilon)| >0 \} \to 0, 
\end{split}
\end{equation*}	
where the last inequality is due to \eqref{wichtig},\eqref{wichtig2} and that, due the very beginning of the proof, $\Psi_n(\theta) \stackrel{p}{\to} \Psi(\theta)$ for $\theta \in \Theta$.
		
		\qed

\section{Proof of Theorem \ref{theolike}} \label{prooftheo1}

First we derive the density of $\mathcal{L}(N^{\ast}_n)$ w.r.t. $\mathcal{L}(N_0)$ to be
\begin{equation} \label{prelikelihood}
	g(\mu)= \left( \prod_{i=1}^{\mu(S)} h_{\theta_0}(x_i,t_i)\right) \exp (A^2 - n \alpha_{\theta_0})
\end{equation} 
and the display \eqref{likelihood} results by replacing the true $\theta_0$ by the generic $\theta$, inserting $h_{\theta}$ from \eqref{lemh} and evaluating at the argument ($\mu$) as the observation $n_n^{\ast}$. 

According to Theorem 3.1.1 in \cite{reiss1993}, it suffices to derive the density only on 
$\mathbb{M}_k:=\{\mu \in \mathbb{M} : \mu(S) =k \}$.

We obtain by \eqref{IMtrunc} and \eqref{nuast} that $\mathcal{L}[(X_i,T_i)']=\nu_n^{\ast}/\nu_n^{\ast}(S)$, and $\mathcal{L}[(X^0_i,T^0_i)']=\nu_0/\nu_0(S)$ by \eqref{nu0}. Both are mappings $\mathcal{B} \rightarrow \mathbb{R}_0^+$ related by a density (being a mapping $S \rightarrow \mathbb{R}_0^+$) 
\[
\frac{d\mathcal{L}[(X_i,T_i)']}{d \mathcal{L}[(X^0_i,T^0_i)']} = \frac{d \nu_n^{\ast}}{d\nu_0}\frac{\nu_0(S)}{\nu_n^{\ast}(S)} = h_{\theta_0}(x_i,t_i) \frac{\nu_0(S)}{\nu_n^{\ast}(S)} = \frac{A^2}{n \alpha_{\theta_0}} h_{\theta_0}(x_i,t_i)=: h_1(x_i,t_i).
\]
That $\nu_0(S)$ and $\nu_n^{\ast}(S)$ are constants leads to the first equality, the second equality is due to \eqref{defh} and the third holds by \eqref{IMtrunc} and \eqref{nu0}.		
The product experiment $\mathcal{L}[(X_i,T_i)']^k$ has $\mathcal{L}[(X^0_i,T^0_i)']^k$-density,    
\begin{equation} \label{defh1k}
	h_{1,k}\left[\binom{x_1}{t_1}, \ldots, \binom{x_k}{t_k} \right] := \prod_{i=1}^k h_1(x_i,t_i)= \frac{A^{2k}}{n^k \alpha_{\theta_0}^k} \prod_{i=1}^k h_{\theta_0}(x_i,t_i)
\end{equation}
Define $\iota_k: S^k \rightarrow \mathbb{M}_k$ with
\[
\iota_k\left[\binom{x_1}{t_1}, \ldots, \binom{x_k}{t_k} \right]:= \sum_{i=1}^k \epsilon_{\binom{x_i}{t_i}},
\]
so that  $h_{1,k}= f_k \circ \iota_k$ with 
\[
f_k(\mu)=h_{1,k}\left[\binom{x_1}{t_1}, \ldots, \binom{x_k}{t_k} \right] 
\]
and $\mu=\sum_{i=1}^k \epsilon_{(x_i, t_i)'}$. The seemingly double-used $h_{1,k}$ represents two {\em different} mappings, due to the different domains ($S^k$ in \eqref{defh1k} and $\mathbb{M}_k$ later). This means that $f_k$ attributes for point measure $\mu$, build on $(x_1,t_1)', \ldots, (x_k,t_k)'$, the same value as $h_{1,k}$ does for the vector $((x_1,t_1)', \ldots, (x_k,t_k)')$.	Now note that for $M \in \mathcal{M}_k$ (with $\mathcal{M}_k$ being the restriction of $\mathcal{M}$ to $\mathbb{M}_k$)
\[
\iota_k \left(\left[ \mathcal{L}\binom{X_i}{T_i} \right]^k \right)(M) = \left[ \mathcal{L}\binom{X_i}{T_i} \right]^k(\iota_k^{-1}(M))=\mathcal{L}\left(\sum_{i=1}^k \epsilon_{\binom{X_i}{T_i}} \right)(M).
\]
It is easiest to start reading the line from the centre, where $\iota_k^{-1}(M)$ is short for $\{ \iota_k^{-1}(\mu), \mu \in M\}$. (Notation to be distinguished from sample size.)
Similarly, $\iota_k [\mathcal{L}(X_i^0,T_i^0)]^k = \mathcal{L}(\sum_{i=1}^k \epsilon_{(X^0_i,T^0_i)'})$. Hence by Lemma 3.1.1 of \cite{reiss1993}, it is $f_k \in d \mathcal{L}\left(\sum_{i=1}^k \epsilon_{(X_i,T_i)'} \right)/d \mathcal{L}\left(\sum_{i=1}^k \epsilon_{(X^0_i,T^0_i)'}\right)$. For $M \in \mathcal{M}_k$,
\begin{equation*}
	P\{ N_n^{\ast} \in M\}  =  P \left\{ \sum_{i=1}^k \epsilon_{\binom{X_i}{T_i}} \in M, Z =k \right\} 
	=  P \left\{ \sum_{i=1}^k \epsilon_{\binom{X_i}{T_i}} \in M \right\} P\{Z =k \}. 
\end{equation*} 
In the first equality, the second condition, $Z =k$, results from the fact that whatever $\mu$, it must be in $\mathcal{M}_k$. For the first condition, the largest index for summation is originally  $Z$, but can be replaced by $k$ due to the second condition. (The order of conditions is irrelevant.) The second equality is due to the independence (see Definitions \ref{defnnast}).
Similarly by Definitions \ref{defn0} for $N_0$: 
\begin{equation} \label{nozerlegung}
	\mathcal{L}\left(N_0\right)(M)=	P\{ N_0 \in M\} 
	=  P \left\{ \sum_{i=1}^k \epsilon_{\binom{X^0_i}{T^0_i}} \in M \right\} P\{Z_0 =k \} 
\end{equation} 
Hence, 
\begin{equation} \label{displaydist}
	\begin{split} 
		P\{ N_n^{\ast} \in M\} & = P\{Z =k \} \int_M  d \mathcal{L}\left(\sum_{i=1}^k \epsilon_{\binom{X_i}{T_i}}\right)  \\
		= & P\{Z =k \}  \int_M f_k d \mathcal{L}\left(\sum_{i=1}^k \epsilon_{\binom{X^0_i}{T^0_i}}\right) 
		=  \frac{P\{Z =k \}}{P\{Z_0 =k \}} \int_M f_k d \mathcal{L}\left(N_0\right)  
	\end{split}
\end{equation}
The last equality is due to \eqref{nozerlegung}. 
Now, due to Definitions \ref{defnnast} and \ref{defn0},  \eqref{nu0}(right), \eqref{nuast}(right) and \eqref{IMtrunc} we have 
\begin{eqnarray*}
	P\{Z =k \} & = & \frac{n^k \alpha_{\theta_0}^k e^{-n \alpha_{\theta_0}}}{k!} \; \textit{and} \; P\{Z_0 =k \} =  \frac{A^{2k} e^{-2A}}{k!}, \\
	\nu_0(S) & = & E[N_0(S)]=E(Z_0) = A^2 \; \textit{and}\\
	\nu_{n,D}(S) & = & E[N_{n,D}(S)] = E(Z)= n \alpha_{\theta_0}.
\end{eqnarray*}
So that
\begin{eqnarray*}
	\frac{P\{Z =k \}}{P\{Z_0 =k \}} & = & \frac{n^k \alpha_{\theta_0}^k e^{-n \alpha_{\theta_0}}}{A^{2k} e^{-2A}} \\
	=  \frac{n^k \alpha_{\theta_0}^k}{A^{2k}} \exp(A^2 - n \alpha_{\theta_0}) & = & \frac{n^k \alpha_{\theta_0}^k}{A^{2k}}\exp[\nu_0(S)- \nu_{n,D}(S)].
\end{eqnarray*}

Hence, by the display \eqref{displaydist} of the distribution  of $\mathcal{L}(N_n^{\ast})$, its density is, inserting \eqref{defh1k},
\begin{equation} \label{stern} 
	\frac{f_k P\{Z =k \}}{P\{Z_0 =k \}}  =  \left( \prod_{i=1}^k h_{\theta_0}(x_i,t_i) \right) \exp[\nu_0(S)- \nu_{n,D}(S)] 
\end{equation}
for $\mu= \sum_{i=1}^k \epsilon_{(X_i,T_i)'}$.

Concluding from $k$ to $\mu(S)$ and inserting the above displays, $\mathcal{L}(N^{\ast}_n)$ (or more informally $N^{\ast}_n$) has $\mathcal{L}(N_0)$-density \eqref{prelikelihood} \cite[see][Theorem 3.1.1 and Example 3.1.1]{reiss1993}
\qed

	\end{appendix}
\end{document}